\documentclass[11pt,a4paper]{article}
\pdfoutput=1
\usepackage{lmodern} 
\usepackage{fullpage}
\usepackage[T1]{fontenc}
\usepackage{lmodern}
\usepackage[tbtags]{mathtools}
\usepackage{amssymb}
\usepackage{amsfonts}
\usepackage{amsthm}
\usepackage{thm-restate}
\usepackage{hyperref}
\hypersetup{colorlinks={true},linkcolor={blue},citecolor=red}

\newtheorem{theorem}{Theorem}

\newtheorem{lemma}{Lemma}
\newtheorem{remark}{Remark}

\theoremstyle{definition}
\newtheorem{definition}{Definition}

\newcommand{\sset}[1]{\left\{ #1\right\}}
\newcommand{\ssets}[1]{\{ #1\}}
\newcommand{\fwh}[1]{\; \left| \; #1 \right.}
\newcommand{\card}[1]{\left| #1 \right|}
\newcommand{\prob}[1]{\ensuremath{\mathrm{Pr}\left[#1\right]}}
\DeclarePairedDelimiter{\ceil}{\lceil}{\rceil}
\newcommand{\union}{\cup}
\DeclareMathOperator*{\expectation}{\mathbb E}
\newcommand{\expect}[1]{\ensuremath{\expectation\left[#1\right]}}
\newcommand{\expectsmall}[1]{\ensuremath{\expectation [#1]}}
\newcommand{\vecc}[1]{\ensuremath{\mathbf{#1}}}
\newcommand{\opt}{\ensuremath{\mathrm{OPT}}}
\newcommand{\fifo}{{\mathcal{M}_{on}}}
	\newcommand{\optm}{{\mathcal{M}}}
\title{Online Market Intermediation\thanks{Supported by the ERC Advanced 
Grant 321171 (ALGAME) and by EPSRC (award reference 1493310).}}

\author{
		Yiannis Giannakopoulos
		\thanks{Department of Informatics, 
			Technical University of Munich. Email: 
			\texttt{giannako@in.tum.de}} 
	\and 
		Elias Koutsoupias
		\thanks{Department of Computer Science, University of Oxford. Email: 
			\texttt{\{elias,filippos.lazos\}@cs.ox.ac.uk}}
	\and
		Philip Lazos\footnotemark[3]
}
\date{March 17, 2017}
\begin{document}

\maketitle

\begin{abstract}
	We study a dynamic market setting where an intermediary interacts with 
	an unknown
	large sequence of agents that can be either sellers or buyers: their 
	identities, as
	well as the sequence length $n$, are decided in an adversarial, online 
	way. 
	Each
	agent is interested in trading a single item, and all items in the market are
	identical. The intermediary has some prior, incomplete knowledge of the 
	agents'
	values for the items: all seller values are independently drawn from the 
	same
	distribution $F_S$, and all buyer values from $F_B$. The two 
	distributions 
	may
	differ, and we make standard regularity assumptions, namely that $F_B$ 
	is 
	MHR and
	$F_S$ is log-concave.
	
	We focus on online, posted-price mechanisms, and analyse two
	objectives: that of maximizing the intermediary's profit and that of 
	maximizing the
	social welfare, under a competitive analysis benchmark. First, on the 
	negative side,
	for general agent sequences we prove tight competitive ratios of
	$\varTheta(\sqrt{n})$ and $\varTheta(\ln n)$, respectively for the two 
	objectives.
	On the other hand, under the extra assumption that the intermediary 
	knows some bound
	$\alpha$ on the ratio between the number of sellers and buyers, we  
	design
	asymptotically optimal online mechanisms with competitive ratios of 
	$1+o(1)$ and
	$4$, respectively. Additionally, we study the model 
	were the number of items that can be stored in stock throughout the 
	execution is bounded, in which case the competitive ratio for the profit is 
	improved to $O(\ln n)$.
\end{abstract}

\section{Introduction}
The design and analysis of electronic markets is of central importance
in algorithmic game theory. Of particular interest are trading
settings, where multiple parties such as buyers, sellers, and
intermediaries exchange goods and money. Typical examples are markets
for trading stocks, commodities, and derivatives: sellers and
buyers where each one trades a single item and one intermediary for
facilitating the transactions. However, the
well-understood cases are comparatively quite modest. The very special 
case of one seller,
and one buyer was thoroughly studied by Myerson and 
Satterthwaite~\cite{myerson1983efficient} in
their seminal paper; they provided a
beautiful characterization of many significant properties a mechanism
might have, along with an impossibility theorem showing that it cannot
possess them all. The paper also dealt with the case where a broker 
provides assistance by making two potential trades,
one with each agent, while also trying to maximize his profit. This
was extended in~\cite{Deng2014} to multiple sellers and buyers that are all 
immediately present in an offline manner.

Our work considers a similar setting, but with a key difference: the buyers 
and sellers
appear one-by-one, in a dynamic way. It is natural to study this question in 
the
incomplete information setting in which the intermediary, whose
objective is to maximize either profit or welfare, does not know the 
sequence of 
buyers and
sellers in advance. The framework that we employ to study the question
is the standard worst-case analysis of online algorithms whose goal is
to do as well as possible in the face of a powerful adversary which
tries to embarrass them.

We are not the first to apply techniques from online algorithms
to quantify uncertainty in markets: the closest work to ours would be by 
Blum et al.~\cite{blum_online_2006}
who consider buyers and
sellers trading identical items. In their setting, motivated mostly
from a financial standpoint, buyers and sellers arrived in an online
manner, with their bids appearing to the trader and expiring after
some time. The trader would have to match prospective buyers and
sellers to facilitate trade. Among a plethora of interesting
results, the trader's profit maximization problem was studied using
competitive analysis and techniques from online weighted matchings.
The key difference in our setting is that buyers and sellers do not
overlap: whenever a seller appears, the intermediary has to decide
whether or not to attempt to buy the item, without having a buyer ready
to go. Instead, the intermediary stores the item to sell it at a later
time. We believe this variation is able to capture ``slower'' markets,
like online marketplaces similar to Amazon or AliExpress (or even
regular retail stores), where uncertainty stems from not knowing how
large a stock of items to buy, in expectation of the buyers to come.

\subsection{Our Results}
Our aim is to study this dynamic market setting, where an intermediary
faces a sequence of potential buyers and sellers in an online
fashion. The goal of the intermediary is to maximize his profit, or society's 
welfare, by
buying from the sellers and selling to buyers. We take a
Bayesian approach to their utilities but use competitive analysis for
their arrivals: the main difficulty stems from the unknown (and
adversarially chosen) sequence of agents. Further particulars and notation 
is discussed in Section~\ref{sec:notation}. All the online algorithms we 
design are posted price, which are simple, robust and strongly truthful.

First, in Section~\ref{sec:general_setting} we study the case of
arbitrary sequences of buyers and sellers and show that the
competitive ratio---the ratio of the optimal offline profit over the
profit obtained by the online algorithm---is $\Theta(\sqrt n)$, where
$n$ is the total number of buyers and sellers.  
We also study the social welfare objective, where the goal is to maximize 
the total
utility of all participants, including the sellers, the buyers and the
intermediary The competitive ratio here is $\varTheta(\log n)$. All these 
results
are achieved via standard regularity assumptions on the distributions of the 
agent
values (see Section~\ref{sec:distros}), which we also prove to be necessary, 
by providing arbitrarily bad
competitive ratios in the case they are dropped 
(Theorem~\ref{th:welfarelower2}). 

To overcome the above pessimistic results, 
we next study in Section~\ref{sec:stocklimit} the setting where both the 
online and offline algorithms have a limited stock, i.e.\ at no point in time 
can they hold more than $K$ items. In this model, the competitive ratio is 
improved to $\varTheta(K\log n)$, asymptotically matching that of welfare. 
Finally, we also propose a way to restrict the input sequence, by 
introducing in Section~\ref{sec:balanced} the notion of $\alpha$-balanced 
streams, where at every prefix of the stream the ratio of the number of 
sellers to buyers has to be at least $\alpha$. Under this condition 
we are able to bring down the competitive ratios for both objectives to 
constants. In particular, the online posted-price mechanism that we use for 
profit maximization, and which is derived by a fractional relaxation of the 
optimal offline profit, achieves an asymptotically optimal ratio of $1+o(1)$. 
A similar mechanism is $4$-competitive for the welfare objective.

\subsection{Prior Work}

Our work is grounded on a string of fruitful research in mechanism
design. The main topics that are close to our effort are bilateral
trading, trading markets and sequential (online) auctions. 

The first step in bilateral trading and mechanism design was made by 
Myerson and 
Satterthwaite~\cite{myerson1983efficient} who proved
their famous impossibility result, even for the case of one buyer and
one seller. The case for profit maximization was extended to many buyers 
and sellers, each trading
a single identical item, in \cite{Deng2014}. Some of the assumptions
in our model are based in these two works. The impossibility result in 
\cite{myerson1983efficient}, among other difficulties, slowly vanishes for 
larger markets as was shown by McAffee~\cite{mcafee_dominant_1992}. 
There is still active
progress being made on this intriguing setting, concentrating on
simple mechanisms that provide good approximations either to welfare 
while
staying budget balanced and individually rational
\cite{blumrosen_almost_2016,blumrosen_approximating_2016} or to 
profit~\cite{niazadeh2014simple}.  Other 
recent developments include a hardness result for computing optimal 
prices~\cite{gerstgrasser2016revenue} and constant efficiency 
approximation with strong budget balance~\cite{colini2016approximately}.

Sequential auctions have also produced a collection of interesting
results, either extending the ideas of simple approximate mechanisms 
instead
of more complex, theoretically optimal ones or dealing with entirely
new settings. Prominent examples that compare the revenue (or welfare) 
generated by simple, posted-price
sequential auctions to the optimal, proving good approximations in certain
cases, are
\cite{Blumrosen2008} for single-item revenue, 
\cite{chawla_multi-parameter_2010,Yan2011} for 
matroid constraints (and some multi-dimensional settings) and 
\cite{feldman2015combinatorial} for combinatorial auctions.
There have been many approaches that apply competitive (worst-case)
analysis to mechanism design. The analysis of competitive auctions for
digital goods is explored in \cite{Bar-Yossef:2002_online,
	blum_near-optimal_2005} where near optimal algorithms are developed
using techniques inspired from no-regret learning. There is also a
deep connection between secretary problems and online sequential
auctions
\cite{hajiaghayi_adaptive_2004,hajiaghayi_online_2005,Babaioff_secretary_immorlica}.
Hajiaghayi et al. utilized techniques such as prophet inequalities  for 
unknown market size with distributional assumptions 
in~\cite{hajiaghayi2007automated}.
A comprehensive exposition of online mechanism design by Parkes can be
found in~\cite{Nisan:2007:AGT}.

There are also positive results in online auctions when the valuation
distribution is unknown (but usually known to be restricted in some
way, having bounded support or being monotone hazard-rate
etc). Babaioff et al.\ explored the case of selling a single item to
multiple i.i.d.\ buyers in \cite{Babaioff2011a}. The case of
$k$ items in a similar setting was studied in \cite{BabaioffDKS15},
while the case of unlimited items (digital goods auctions) in
\cite{Kleinberg:2003_online_learning} and
\cite{koutsoupias2013competitive}. Budget constraints where also
introduced in \cite{badanidiyuru_learning_2012}, where a procurement
auction was the focus.

\section{Preliminaries and Notation}
\label{sec:notation}
The input is a finite string $\sigma\in\ssets{S,B}^*$ of buyers ($B$) and 
sellers
($S$). The online algorithm has no knowledge of $\sigma(t)$, i.e.\ whether
$\sigma(t)=S$ or $\sigma(t)=B$, before step $t$. Also, it doesn't know the 
length
$n(\sigma)$ of $\sigma$. Denote $n_S(\sigma)$, $n_B(\sigma)$ the 
number of sellers
and buyers, respectively, in $\sigma$, and let $N_S(\sigma)$, 
$N_B(\sigma)$ be the
corresponding set of indices, i.e.\ 
$N_S(\sigma)=\ssets{t\fwh{\sigma(t)=S}}$ and
$N_B(\sigma)=\ssets{t\fwh{\sigma(t)=B}}$. Let 
$N(\sigma)=N_S(\sigma)\union
N_B(\sigma)=\ssets{1,2,\dots,n(\sigma)}$. In the above notation we will 
often drop
the $\sigma$ if it is clear which input stream we are referring to.

The values of the sellers are drawn i.i.d.\ from a probability distribution 
(with 
cdf) $F_S$ and these of buyers i.i.d.\ from a distribution $F_B$, both 
supported over intervals of nonnegative reals. 
We denote the random variable of the value of the $t$-th agent 
with $X_t$. We assume that distributions $F_S$ and $F_B$ are continuous, 
with bounded expectations $\mu_S$ and $\mu_B$, and have (well-defined) 
density functions $f_S$ and $f_B$, respectively. It will also be useful to 
denote by $X_S$ a random variable drawn from distribution $F_S$, and 
similarly $X_B\sim F_B$, and for any random variable $Y$ and positive 
integer $m$ use $Y^{(m)}$ to represent the maximum order statistic out of 
$m$ i.i.d.\ draws from the same distribution as $Y$. We will also use the 
shortcut notation $\mu^{(m)}=\expectsmall{Y^{(m)}}$.

We study \emph{posted-price} online algorithms that upon seeing the 
identity of 
the $t$-th agent (whether she is  a seller or a buyer), offer a price $p_t$. We 
buy one unit of the item from sellers that accept our price (i.e.\ if 
$\sigma(t)=S$ and $X_t\leq p_t$) and pay them that price, and we sell to 
buyers that accept our price (i.e.\ if $\sigma(t)=B$ and $X_t\geq p_t$), 
given stock availability (see below), and collect from them that price. So, a 
price 
$p_{t+1}$ can only depend on $\sigma(1),\dots,\sigma(t+1)$ and the 
result of the comparison $X_i\leq p_i$ in all previous steps 
$i=1,2,\dots,t$. Let $K_t$ denote the available stock at the 
beginning of the $t$-th step, i.e.\ $K_1=0$ and 
\begin{equation*}
K_{t+1}=
\begin{cases}
K_{t}+1, &\text{if}\;\; \sigma(t)=S \; \land \; X_{t}\leq p_{t}\\
K_{t}-1, &\text{if}\;\; \sigma(t)=B \; \land \; K_{t}\neq 0 \; \land \; 
X_{t}\geq p_{t}\\
K_{t}, &\text{otherwise}.
\end{cases}
\end{equation*}
Then, the set of sellers from whom we bought items during the algorithm's 
execution is  $I_S=\sset{t\in N_S\fwh{X_t\leq p_t}}$ and the set of buyers 
we sold to is $I_B=\sset{t\in 
	N_B\fwh{X_t\geq p_t \land K_t\neq 0}}$. Notice that these are random 
	variables, depending on the actual realizations of the agent values $X_t$.

The total \emph{profit} that the intermediary deploying an algorithm $A$ 
makes throughout the execution on an input stream $\sigma$, is the 
amount he manages to collect from the buyers via successful sales, minus 
the amount he spent in order to maintain stock availability from the sellers, 
that is
\begin{equation*}
\mathcal R(A,\sigma)=\expect{\sum_{t\in I_B} p_t- \sum_{t\in I_S} p_t }.
\end{equation*}

The social \emph{welfare} of  algorithm $A$ is the sum of valuations that all
participants achieve throughout the entire execution. That is, a seller at 
position
$t$ of the stream has a value of $X_t$ if she keeps her item, or a value of 
$p_t$ if
she sold the item to the intermediary; a buyer has a value of
$X_t-p_t$ if she managed to buy an item, since the item has a value of 
$X_t$ and he
spent $p_t$ to buy it, or $0$ otherwise. And the intermediary, has a value of
$\mathcal{R}(A)$ plus the value of the items that he didn't manage to sell in 
the
end and which are now left in his stock. Putting everything together and 
performing the
occurring cancellations, this results in the welfare to be expressed simply 
as the sum of the
values of the sellers that kept their items plus the sum of the values of the 
buyers
that bought an item, i.e.\
\begin{equation}
\label{eq:welfare1}
\mathcal W(A,\sigma)=\expect{\sum_{t\in N_S\setminus I_S} X_t + 
\sum_{t\in I_B} X_t}.
\end{equation}

We use \emph{competitive analysis}, the standard benchmark for online 
algorithms (see e.g.~\cite{Borodin1998a}), in order to quantify the 
performance of an online algorithm $A$: we compare it to that of an 
unrealistic, offline optimal algorithm $\opt$ has access to the entire 
stream $\sigma$ in advance. Then, we say that $A$ is 
$\rho(n)$-competitive with respect to welfare, if for any feasible input 
sequence of agents $\sigma$ with length $n$ and distributions $F_S$, 
$F_B$ for the agent values, it is
$\mathcal{W}(\opt,\sigma)\leq \rho(n)\cdot \mathcal{W}(A,\sigma)$. 
Notice how we allow the competitive ratio $\rho(n)$ to explicitly depend on 
the input's length, so that we can perform asymptotic analysis as 
$\mathcal{W}(\opt,\sigma)$ and $n$ tend to infinity. It is common in 
competitive analysis to allow for an additional constant in the right hand 
side of the above expression, that does not depend in the input, and which 
intuitively can capture some initial configuration disadvantage of the online 
algorithm. We do that for the case of the profit objective, as this constant 
will have a very natural interpretation: you can think of it as the maximum 
amount of deficit on which an online algorithm can run at any point in time, 
since an adversary can always stop the execution at any time he wishes. 
Given that interpretation, it makes sense to allow for this constant to 
depend on seller distribution $F_S$, since even when we face a single seller 
at the first step we expect to spend an amount that depends on the 
realization of her value. Thus, we will say that an online algorithm is 
$\rho(n)$-competitive with respect to welfare, if for any input sequence of 
agents $\sigma$ and any probability priors $F_S,F_B$, 

\begin{equation}
\label{eq:crprofitdef}
\mathcal{R}(\opt,\sigma)\leq \rho(n)\cdot \mathcal{R}(A,\sigma)+O(\mu_S).
\end{equation}

\section{Distributional Assumptions}
\label{sec:distros}
Throughout most of the paper we will make some assumptions on the 
distributions $F_B$, $F_S$ from which the buyer and seller values are drawn.
In particular, we will assume that $F_B$ has \emph{monotone hazard rate 
(MHR)}, i.e.\ $\log(1-F_B(x))$ is concave, and that $F_S$ is 
\emph{log-concave}, i.e.\ $\log F_S(x)$ is concave. 
For convenience, we will collectively refer to both the above constraints as 
\emph
{regularity}.
These conditions are rather standard in the optimal auctions literature, and 
they encompass a large class of natural of distributions including  e.g.\  
exponential, uniform and normal ones. 
Notice that distributions that satisfy the above conditions also fulfil the 
regularity requirements introduced in the seminal paper Myerson and 
Satterthwaite~\cite{myerson1983efficient} for the single-shot, one buyer 
and one seller setting of bilateral trade, namely that 
$x+\frac{F_S(x)}{f_S(x)}$ and $x-\frac{1-F_B(x)}{f_B(x)}$ are both increasing 
functions. 
Finally, we must mention that such regularity assumptions are necessary, in 
the sense that dropping them would result in arbitrarily bad lower bounds 
for the competitive ratios of our objectives, as it is demonstrated by 
Theorem~\ref{th:welfarelower2}. 

The following two lemmas demonstrate some key properties of the regular 
distributions that will be very useful in our subsequent analysis:

\begin{theorem}
	\label{th:MHRprops}
	For any random variable $Y$ drawn from an MHR distribution with 
	bounded expectation
	$\mu$ and standard deviation $s$,
	\begin{enumerate}
		\item \label{prop:MHR1} $\prob{Y\geq y}\geq \frac{1}{e}$ for any 
		$y\leq \mu$
		\item \label{prop:MHR12} $\prob{Y\geq y}< \frac{1}{e}$ for any $y> 
		2\mu$
		\item \label{prop:MHR2} $\expectsmall{Y^{(m)}}\leq H_m\cdot \mu $, 
		where $H_m$ is
		the $m$-th harmonic number.
		\item \label{prop:MHR3} $s \leq \mu$
	\end{enumerate}
\end{theorem}
\begin{proof} A proof of Property \ref{prop:MHR1} can be found in 
	\cite[Theorem
	3.8]{Barlow1964}, of Property \ref{prop:MHR12} in \cite[Corollary 
	3.10]{Barlow1964}, and of Property~\ref{prop:MHR2} in \cite[Lemma 
	13]{Babaioff2011a}.
	For Property~\ref{prop:MHR3}, from \cite[Lemma 2]{gkyr2015-wine} we 
	know that $\expect{Y^2}\leq 2\mu^2$, so 
	$s^2=\expect{Y^2}-\mu^2\leq 
	\mu^2$.
\end{proof}

\begin{lemma}\label{lemma:tailconcavedistro}
	For any distribution over $[0,\infty)$ with log-concave cdf $F$ and 
	expectation $\mu$, 
	$$
	x\leq e\mu F(x)\qquad \text{for any}\;\; x\leq \mu.
	$$
\end{lemma}
\begin{proof}
	Fix some $x\leq \mu$ and let $c=\frac{x}{\mu}$. Define the random 
	variable $Y=cX$, where $X$ is drawn from $F$, and let $F_Y$ be the cdf 
	of $Y$. Since $F$ is log-concave, $\ln F(t)$ is a concave function, and so 
	from Jensen's inequality
	$$
	\ln F(c\mu)=\ln F(\expectsmall{Y})\geq \int_0^\infty \ln F(t)\,dF_Y(t)= 
	\int_0^\infty \ln F(t)c\,dF(t)=c \int_0^1\ln u\,du=-c.
	$$
	So, $F(x)\geq 
	e^{-c}=\frac{c\mu}{\mu}\frac{e^{-c}}{c}=\frac{x}{\mu}\frac{e^{-c}}{c}$. 
	The lemma follows from the fact that $\frac{e^{-c}}{c}$ is decreasing for 
	$c\in (0,1]$.
\end{proof}

Finally, we prove the following property bounding the sum of maximum 
order statistics of a distribution, that holds for general (not necessarily 
regular) distributions and might be of independent interest:

\begin{lemma}
	\label{lemma:higherorderstats}
	The expected average of the $k$-th highest out of $m$ independent 
	draws 
	from a probability distribution with expectation $\mu$ and standard 
	deviation $s$ can be at most $\mu+2\sqrt{\frac{m}{k}}s$.
\end{lemma}
\begin{proof}
	Let $Y^{(1:m)}\leq Y^{(2:m)} \leq Y^{(m:m)}$ denote the order statistics of 
	$m$ independent draws from a probability distribution with mean $\mu$ 
	and standard deviation $s$. We want to prove that 
	$$\sum_{i=m-k+1}^m\expectsmall{Y^{i:m}}\leq k\mu+2\sqrt{km} s.$$
	
	From \cite[Eq.~(4)]{Arnold1979} we know that 
	$\expectsmall{Y^{(i:m)}}\leq 
	\mu +s \sqrt{\frac{i-1}{m-i+1}}$, so it is enough to show that 
	$\sum_{i=m-k+1}^m\sqrt{\frac{i-1}{m-i+1}}\leq 2\sqrt{km}.$ Indeed, 
	by 
	using the transformation $j=m-i+1$, we get
	$$
	\sum_{i=m-k+1}^m\sqrt{\frac{i-1}{m-i+1}}=\sum_{j=1}^k\sqrt{\frac{m}{j}-1}\leq
	\sqrt{m}\sum_{j=1}^k\sqrt{\frac{1}{j}}\leq \sqrt{m}\int_{0}^k 
	x^{-1/2}\,dx=\sqrt{m}\cdot 2\sqrt{k}.
	$$
\end{proof}

\section{General Setting}
\label{sec:general_setting}
We start by studying the general setting where no additional assumptions 
are enforced on the structure of the input sequence. The adversary is free to 
arbitrarily choose the identities of the agents.
\subsection{Welfare}
\begin{theorem}
	\label{th:welfareupper2}
	For regularly distributed agent values\footnote{As matter of fact, in the 
		proof of Theorem~\ref{th:welfareupper2} just regularity for the buyer 
		values would suffice, i.e.\ $F_B$ being MHR.}, the online auction that 
	posts to any seller and buyer the median of their distribution is $O(\ln 
	n)$-competitive with respect to welfare. This bound is tight.
\end{theorem}
\begin{proof}
	We split the proof of the theorem in two more general lemmas below, 
	corresponding to upper and lower bounds. Then, the upper bound for 
	our case of regular distributions follows easily from 
	Lemma~\ref{th:welfareupper1} by using constants $c_1=c_2=2$, and 
	taking into consideration that, from 
	Property~\ref{prop:MHR2} of
	Theorem~\ref{th:MHRprops}, the ratio of the maximum order statistic for 
	the MHR 
	distribution $F_B$ is upper bounded by $r_B(m)\leq H_m\leq O(\ln m)$. 
	For the lower bound, it is enough to observe that this ratio is attained by 
	an exponential distribution, which is MHR.

	\begin{lemma}
		\label{th:welfareupper1}
		For any choice of constants $c_1,c_2> 1$, the following fixed-price 
		online auction has a competitive ratio of at most
		$\max\sset{\frac{c_1}{c_1-1},c_1c_2\cdot r_B(n_B)}$ with respect to 
		welfare, where $n_B$ is the number of buyers, and 
		$r_B(m)=\mu^{(m)}_B/\mu_B$ is the ratio between the 
		$m$-maximum-order statistic and the expectation of the buyer value 
		distribution.
		\begin{itemize}
			\item Post to all sellers price $q=F^{-1}_S\left(\frac{1}{c_1}\right)$.
			\item Post to all buyers price 
			$p=F^{-1}_B\left(\frac{c_2-1}{c_2}\right)$.
		\end{itemize}
	\end{lemma}

	\begin{proof}
		Let $A$ denote our online algorithm and $\opt$ an offline algorithm 
		with 
		optimal
		expected welfare.  Fix an input stream $\sigma$. Looking at 
		\eqref{eq:welfare1}, the
		maximum welfare that $\opt$ can get from the sellers is at most 
		$\expect{\sum_{t\in
				N_S}X_t}=n_s\mu_S$, while from the buyers at most 
		$\expect{\card{I_B}\cdot
			X_B^{(n_B)}}\leq \kappa \expect{X_B^{(n_B)}}$, where $\kappa$ is 
			the 
		maximum number
		of sellers that can be matched to \emph{distinct} buyers that arrive 
		after
		them\footnote{You can think of that as the maximum size of a 
		matching 
			in the
			following undirected graph: the nodes are the sellers and the 
			buyers, 
			and there is
			an edge between any seller and all the buyers that appear after her 
			in 
			$\sigma$.} in
		$\sigma$: clearly, no mechanism can sell more than $\kappa$ items. 
		Bringing all together we have that
		$$
		\mathcal{W}(\opt)\leq n_s\mu_S + \kappa \mu^{(n_B)}_B=n_s\mu_S 
		+r_B(n_B)\cdot \kappa \mu_B.
		$$ 
		For the online algorithm now, from the sellers we get 
		$$\sum_{i\in N_S}\prob{X_i>q}\expectsmall{X_i|X_i>q}\geq 
		n_s(1-F_S(q))\expectsmall{X_S}=\frac{c_1-1}{c_1}\cdot n_S\mu_S$$
		and from the buyers at least
		$$
		\kappa \prob{X_S\leq q}\prob{X_B\geq p}\expectsmall{X_i|X_i\geq 
			p}\geq \kappa F_S(q)(1-F_B(p))\expectsmall{X_B}= 
		\frac{1}{c_1} \frac{1}{c_2}\cdot \kappa\mu_B,
		$$
		just by considering one of the $\kappa$-size matchings discussed 
		before: if we manage to buy from one of these $\kappa$ sellers, then 
		we 
		will definitely have stock availability for the matched buyer.
	\end{proof}

	The upper bound in Lemma~\ref{th:welfareupper1} cannot be improved:
	
	\begin{lemma}
		\label{th:welfarelower1}
		For \emph{any} probability distribution $F$,
		even if the seller and buyer values are i.i.d.\ from $F$, 
		the sequence $SB^n$ forces all posted-price online mechanisms to 
		have 
		a competitive ratio of $\varOmega(r(n))$, where 
		$r(n)=\mu^{(n)}/\mu$ is the ratio of the $n$-maximum-order 
		statistic of distribution $F$ to its expectation.
	\end{lemma}
	\begin{proof}
		Assume that the seller and buyer values are drawn i.i.d.\ from a 
		distribution $F$. Let $Y\sim F$ denote a random variable following 
		this 
		distribution and denote $\mu=\expectsmall{Y}$, 
		$\mu^{(n)}=\expectsmall{Y^{(n)}}$. Fix an online algorithm $A$ that 
		posts price $q$ to the seller and prices $\vecc p \equiv 
		p_1,p_2,\dots$ 
		to the buyers. Notice that this sequence of buyer prices $\vecc p$ 
		cannot depend on the actual stream length $n$, since that is being 
		selected adversarially.  
		
		We overestimate $A$'s expected welfare by assuming that it gets 
		maximum welfare from the first seller, i.e.\ $\expectsmall{Y}=\mu$, 
		while at the same buys for sure the item from her so that it has stock 
		availability to sell in the sequence of buyers. 
		Then, from \eqref{eq:welfare1} its expected welfare is given by
		\begin{equation}
		\label{eq:welfareseries}
		\mathcal{W}(\vecc p)=\mu+\sum_{t=1}^n \pi(t)\cdot \lambda(p_t),
		\end{equation}
		where 
		$$\pi(t)=\pi(\vecc 
		p,t)=\prod_{j=1}^{t-1}\prob{Y<p_j}=\prod_{j=1}^{t-1}F(p_j)$$
		and
		$$
		\lambda(y)=\prob{Y\geq y}\cdot \expect{Y\fwh{Y\geq 
				y}}=(1-F(y))\expect{Y\fwh{Y\geq y}}=\int_{ y}^\infty xf(x)\,dx 
				\leq 
		\mu.
		$$
		
		First we show that we can without loss assume that the buyer prices 
		are 
		nonincreasing. Indeed, for a contradiction suppose that exists a time 
		step $t^*$ such that 
		$\alpha\equiv p_{t^*}<p_{t^*+1}\equiv \beta$. Consider now the 
		online 
		mechanism that uses prices $\vecc p'$, where $\vecc p'$ results from 
		the original prices $\vecc p$ if we flip the prices at steps $t^*,t^*+1$, 
		i.e.\ $p_{t^*}'=\beta$, $p_{t^*+1}'=\alpha$, and $p'_{t}=p_{t}$ for all 
		$t\neq t^*,t^*+1$. Then, the difference in the expected welfare 
		between 
		the two mechanisms is
		\begin{align}
		\mathcal W(\vecc p')-\mathcal W(\vecc p) 
		&=\sum_{t=t^*}^{t^*+1} \pi(\vecc p',t)\cdot 
		\lambda(p_t')-\sum_{t=t^*}^{t^*+1} \pi(t)\cdot \lambda(p_t) \notag\\
		&= \pi(t^*)\lambda(\beta)+\pi(t^*)F(\beta)\lambda(\alpha)
		- \pi(t^*)\lambda(\alpha)-\pi(t^*)F(\alpha)\lambda(\beta) \notag\\
		&= 
		\pi(t^*)\left[(1-F(\alpha))\lambda(\beta)-(1-F(\beta))\lambda(\alpha) 
		\right] \notag\\
		&= \pi(t^*)(1-F(\alpha))(1-F(\beta))\left(\expect{Y\fwh{Y\geq \beta}} - 
		\expect{Y\fwh{Y\geq \alpha}} \right), \label{eq:helper2} 
		\end{align}
		which is nonnegative since $\alpha<\beta$.
		
		There are two options for the prices $\vecc p$: either $F(p_t)=1$ 
		for all $t$, or $k=\min\sset{t\fwh{F(p_t)<1}}$ is a well-defined 
		positive 
		integer that does not depend on $n$, in which case define the 
		constant 
		$c\equiv F(p_k)<1$. From \eqref{eq:welfareseries}, in the former case 
		it 
		is easy to see that $W(\vecc p)= \mu $, while in the latter one 
		$$
		\mathcal W(\vecc p)
		\leq \mu +\pi(k)\sum_{t=k}^nF(p_k)^{t-k}\lambda(p_t)
		\leq  \mu +\pi(k)\sum_{t=k}^n c^{t-k}\mu
		\leq 
		\left(1+\sum_{j=0}^\infty c^j \right)\mu
		=\frac{2-c}{1-c}\mu
		$$ 
		
		On the other hand, it is a well-know fact from the theory of prophet 
		inequalities (see e.g.~\cite{Kleinberg:2012}) that by using a price of 
		$\frac{\mu^{(n)}}{2}$ for all the buyers an offline mechanism can 
		achieve 
		a welfare of at least $\frac{\mu^{(n)}}{2}$ from the buyers, given of 
		course availability of stock. So, by setting e.g.\ a price equal to the 
		median of $F$ for the seller, the optimal offline welfare is at least 
		$\frac{1}{2}\mu +\frac{1}{4}\mu^{(n)}=\varOmega(\mu^{(n)})$.
	\end{proof}
\end{proof}

As the following theorem demonstrates, the regularity assumption on the 
agent values is necessary if we want to hope for non-trivial bounds. In 
particular, the lower bound in Lemma~\ref{th:welfarelower1} can be made 
arbitrarily high:

\begin{theorem}
	\label{th:welfarelower2}
	For any constant $\varepsilon\in (0,1)$, there exists a continuous 
	probability 
	distribution $F$ such that any online posted-price mechanism has a 
	competitive ratio of $\varOmega(n^{1-\varepsilon})$ on the input 
	sequence $SB^n$, even if the values of the sellers and the buyers are 
	i.i.d.
\end{theorem}
\begin{proof}
	Fix some $\varepsilon \in (0,1)$ and choose the Pareto distribution with 
	$F(x)=1-x^{-\frac{1}{1-\varepsilon}}$ for $x\in [1,\infty)$.
	The expected value of this distribution is $\mu=\frac{1}{\varepsilon}$ 
	while the expectation of the maximum order statistic out of $n$ 
	independent draws is
	$$
	\mu^{(n)}=\frac{n\varGamma(n)\varGamma(\varepsilon)}{\varGamma(n+\varepsilon)}\sim
	\varGamma(\varepsilon) n^{1-\varepsilon},
	$$
	since 
	$\lim_{n\to\infty}\frac{\varGamma(n+\varepsilon)/\varGamma(n)}{n^{\varepsilon}}=1$,
	where $\varGamma(x)$ denotes the standard gamma function.
	So, as $n$ grows large, the ratio in Lemma~\ref{th:welfarelower1} 
	becomes 
	$$r(n)=\frac{\mu^{(n)}}{\mu}=\varepsilon\varGamma(\varepsilon)\cdot 
	n^{1-\varepsilon}\geq \frac{4}{5}
	n^{1-\varepsilon}=\varOmega(n^{1-\varepsilon}).$$
\end{proof}
\subsection{Profit}

Now we turn our attention to our other objective of interest, that of 
maximizing the expected profit of the intermediary. As it turns out, this 
objective has some additional challenges that we need to address. For 
example,
as the following  theorem demonstrates, if the distribution of seller values 
is bounded away from $0$, the competitive ratio can be arbitrarily bad, 
even for i.i.d.\ values from a uniform distribution:

\begin{theorem}\label{th:lowerrevenue1} 
	For any $a>0$ and $\varepsilon\in (0,1)$, if the seller and
	buyer values are drawn i.i.d.\ from the uniform distribution over $[a,b]$ 
	where $b>
	2a$, then no online posted-price mechanism can have an approximation 
	ratio better
	than $a\left(1-\frac{1}{k}\right)^4 n^{1-\varepsilon}$ with respect to 
	profit,
	where $k=\frac{b}{a}-1$. In particular, for any uniform distribution over 
	an
	interval $[1,h]$ with $h\geq 3$ the lower bound is
	$\frac{1}{2^{4}}n^{1-\varepsilon}=\varOmega\left(n^{1-\varepsilon}\right)$.
\end{theorem}
\begin{proof}
	Fix $a,b>0$ such that $k\equiv\frac{b}{a}-1>1$. Assume that the buyer
	and seller values are drawn i.i.d.\ from the uniform distribution $[a,b]$, 
	i.e.\ the
	cdf is $F(x)=\frac{x-a}{b-a}=\frac{x-a}{ak}$ for all $x\in[a,(k+1)a]$. 
	Consider the
	input stream $\sigma=S^{n/2}B^{n/2}$, for $n$ even.
	
	First, it is easy to see that for any $\varepsilon\in (0,1)$ no online 
	algorithm can
	buy more than $\frac{n^\varepsilon}{128}\frac{1}{a}$ items from the 
	sellers in the
	first part of the stream, otherwise it will have to spend more than
	$\frac{n^\varepsilon}{128}=\omega(1)$. This means that the maximum 
	profit 
	that an online
	algorithm can get, even if it manages to sell to the buyers all the items 
	she bought
	from the sellers, is at most
	$\frac{n^\varepsilon}{128}\frac{1}{a}(b-a)=\frac{k}{128}n^\varepsilon$.
	
	Consider an offline algorithm that posts to seller and buyers the prices
	corresponding to the $\frac{1}{8}\left(1- \frac{1}{k}\right)^2$ and
	$\frac{1}{2}\left(1- \frac{1}{k}\right)$ percentiles, respectively. That is, 
	buyers
	get a price of $p=F^{-1}(y)=a(yk+1)$ and sellers
	$q=F^{-1}\left(\frac{y^2}{2}\right)=\frac{a}{2}(2+y^2k)$, where
	$y=\frac{1}{2}\left(1- \frac{1}{k}\right)$. Then, the probability that the 
	offline
	algorithm buys an item from a specific seller is $F(q)$, resulting in the 
	algorithm
	spending $\frac{n}{2}F(q)q$ in expectation. On the other hand, 
	underestimate its
	expected income buy considering only selling to the $i$-th buyer the 
	item that you
	got from the $i$-th seller. Then, the probability of achieving a successful
	transaction with a particular buyer is $F(q)(1-F(p))$, resulting in an 
	expected
	profit of at least
	
	\begin{align*}
	\frac{n}{2}F(q)(1-F(p))p-\frac{n}{2}F(q)q
	&=\frac{n}{2}\frac{y^2}{2}\left[(1-y)F^{-1}(y)-
	F^{-1}\left(\frac{y^2}{2}\right)\right]\\ &= 
	a\frac{n}{8}y^3\left[(3y-2)k-2
	\right]\\ &= \frac{ak}{128}n\left(1-\frac{1}{k}\right)^4.
	\end{align*}
\end{proof}

If we consider distributions supported over intervals that include $0$, under
standard regularity assumptions we can do a little better than the trivial 
lower
bound of Theorem~\ref{th:lowerrevenue1}:

\begin{theorem}
	\label{th:upperrrevenue1} For agent values regularly distributed over 
	intervals that include $0$, the following online
	posted-price mechanism achieves a competitive ratio of
	$O(n^{\frac{1}{2}+\varepsilon})$ for any $\varepsilon>0$:
	\begin{itemize}
		\item Post to the $i$-th seller price 
		$q_i=F_S^{-1}\left(\frac{1}{e}\frac{1}{i^{1/2+\varepsilon}}\right)$
		\item Post to all buyers price $p=\mu_B$.
	\end{itemize}
\end{theorem}

\begin{proof}
	
	Fix an input stream $\sigma$ of length $n$. Let $\mu_B$ and $s_B$ be 
	the expectation and standard deviation of the buyer value distribution 
	$F_B$.
	As in the proof of Lemma~\ref{th:welfareupper1}, let $\kappa$ denote 
	the 
	maximum number of sellers that can be matched to distinct buyers that 
	arrive after them in $\sigma$. If $\mu_{B}^{(j:m)}$ denotes the 
	expectation 
	of the $j$-th largest out of $m$ independent draws from $F_B$, since no 
	algorithm can make more than $\kappa$ sales over its entire execution, 
	the 
	optimal offline profit is upper bounded by 
	$$\sum_{j=1}^{\kappa} \mu_{B}^{(n_B-j+1:n_B)}\leq 
	\sum_{i=n-\kappa+1}^{n} \mu_{B}^{(i:n)} \leq \kappa \mu_B 
	+2\sqrt{\kappa n}s_B\leq 3\sqrt{\kappa}\sqrt{n} \mu_B,$$
	where for the second inequality we have used 
	Lemma~\ref{lemma:higherorderstats} and for the last one we have used 
	Property~\ref{prop:MHR3} from Theorem~\ref{th:MHRprops} and the 
	obvious fact 
	that $\kappa\leq n$.
	
	For the analysis of the online mechanism now, the expected number of 
	items that it gets from the first $\kappa$ sellers is $\sum_{i=1}^\kappa 
	F_S(q_i)=\frac{1}{e}\sum_{i=1}^\kappa\frac{1}{i^{1/2+\varepsilon}}\geq 
	\frac{1}{e}\kappa^{1/2-\varepsilon}$. 
	So, by considering the FIFO matching between these first $\kappa$ 
	sellers 
	and their corresponding buyers (see Lemma~\ref{lemma:fifo}), the 
	expected 
	income of our algorithm is at least $\frac{1}{e}\kappa^{1/2-\varepsilon} 
	(1-F(p))=\frac{1}{e}\kappa^{1/2-\varepsilon} (1-F(\mu_B))\geq 
	\frac{1}{e^2}\kappa^{1/2-\varepsilon}$, where in the last step we 
	deployed 
	Property~\ref{prop:MHR1} of Theorem~\ref{th:MHRprops}. So, it only 
	remains to be 
	shown that the online algorithm does not spend more than a constant 
	amount. Indeed, our expected spending is at most 
	$$\sum_{i=1}^\infty q_iF_S(q_i)\leq \sum_{i=1}^\infty e\mu_S 
	F_S(q_i)^2 = \frac{1}{e}\mu_S\sum_{i=1}^\infty 
	\frac{1}{i^{1+2\varepsilon}} =O(\mu_S),$$
	where for the first inequality we have used 
	Lemma~\ref{lemma:tailconcavedistro}, 
	taking into consideration that seller prices $q_i$ are decreasing and 
	$q_1$ 
	is below $\mu_S$. This is true because again from 
	Lemma~\ref{lemma:tailconcavedistro} for $x=\mu_S$ we know that $\mu_S 
	\leq e\mu_S 
	F(\mu_S)$, or equivalently $F(\mu_S)\geq \frac{1}{e}=F(q_1)$.
\end{proof}

The algorithm of Theorem~\ref{th:upperrrevenue1} is asymptotically 
optimal:

\begin{theorem}
	\label{th:lowerrevenue2}
	If the seller and buyer values are drawn i.i.d.\ from the uniform 
	distribution over $[0,1]$, then no online posted-price mechanism can 
	have an approximation ratio better than $\varOmega\left(\sqrt{n}\right)$.
\end{theorem}

\begin{proof} As in the lower bound proof of 
	Theorem~\ref{th:lowerrevenue1} we again deploy
	an input sequence $\sigma=S^{n/2}B^{n/2}$ with $n$ even. Let $F(x)=x$ 
	be the cdf of
	the uniform distribution over $[0,1]$. This time we argue that no online 
	algorithm
	can buy more than $\varOmega(\sqrt{n})$ items from the sellers, in 
	expectation. Indeed, let $q_i$ be the price that
	the online mechanism posts to the $i$-th seller. Then, the expected 
	number of items
	$m_\sigma$ bought from the sellers is 
	$\sum_{i=1}^{n/2}F(q_i)=\sum_{i=1}^{n/2}q_i$,
	while the expected expenditure $c_\sigma$ is
	$\sum_{i=1}^{n/2}F(q_i)q_i=\sum_{i=1}^{n/2}q_i^2$. By the convexity of 
	the function
	$t\mapsto t^2$ and Jensen's inequality it must be that
	$$
	m_\sigma=\sum_{i=1}^{n/2}q_i \leq
	\sqrt{\frac{n}{2}}\left(\sum_{i=1}^{n/2}q_i^2\right)^{\frac{1}{2}}=O\left(\sqrt{c_\sigma}\sqrt{n}\right),
	$$
	so given that our deficit must be $c_\sigma=O(\frac{1}{2})$, we get the 
	desired
	$m_\sigma=O(\sqrt{n})$. As a result, the online profit can be at most
	$O(\sqrt{n})\cdot 1=O(\sqrt{n})$.
	
	For the offline algorithm we use prices $q=\frac{1}{8}$ and 
	$p=\frac{1}{2}$ for the
	buyers and sellers, respectively, and by an analogous analysis to that of 
	the proof
	of Theorem~\ref{th:lowerrevenue1}, we get that the expected offline 
	profit 
	is at least
	$$
	\frac{n}{2}F(q)(1-F(p))p-\frac{n}{2}F(q)q=\frac{n}{2}\frac{1}{8}\left(1-\frac{1}{2}
	\right)\frac{1}{2}-\frac{n}{2}\frac{1}{8}\frac{1}{8}=\frac{n}{128}=\varOmega(n).
	$$
\end{proof}

\section{Limited Stock}
\label{sec:stocklimit}
If one looks carefully at the lower bound proof for the profit in 
Theorem~\ref{th:lowerrevenue2}, it becomes clear that the source of 
difficulty for any online algorithm is essentially the fact that without 
knowledge of the future, you cannot afford to spend a super-constant 
amount of money into accumulating a  large stock of items, without the 
guarantee that there will be enough demand from future buyers. In 
particular, it may seem that the offline algorithm has an unrealistic 
advantage of using a stock of infinite size. The natural way to mitigate this 
would be to introduce an upper bound $K$ on the number of items that 
both the online and offline algorithms can store at any point in time. As it 
turns out, this has a dramatic improvement in the competitive ratio for the 
profit:
\begin{theorem}
	\label{th:stocklimit} 
	Assuming stock sizes of at most $K$ items, under our standard 
	regularity assumptions the following online mechanism is $O\left(Kr\log 
	n\right)$-competitive, where $r=\max\sset{1,\frac{\mu_S}{\mu_B}}$:
	\begin{itemize}
		\item If your stock is not currently full, post to sellers price 
		$q=F_S^{-1}\left(\frac{1}{r}\frac{1}{2eK}\right)$
		\item Post to all buyers price $p=\mu_B$.
	\end{itemize}
\end{theorem}
\begin{proof}
	The proof is similar to that of Theorem~\ref{th:upperrrevenue1}, but 
	certain points
	need some special care. Let $\kappa$ again be the maximum number of 
	sellers that
	can be matched to distinct buyers that follow them, but this time under 
	the
	added  restriction of the $K$-size stock. This corresponds 
	to the maximum matching with no ``temporal'' cut of size greater than 
	$K$. We write
	``temporal'' cut to mean any cut in the graph that separates the vertices
	(buyers and sellers) $1\ldots i$ from vertices $i+1\ldots n$ --- that is, 
	precisely the
	condition that we cannot match more than $K$ sellers from an initial 
	segment to
	buyers later in the sequence.
	Lemma~\ref{lemma:fifo} in the appendix demonstrates that such a  
	$\kappa$-size matching can be computed not only offline, but also 
	online using a FIFO 
	queue of length $K$, adding sellers to the queue while it is not full and
	matching buyers greedily: we post prices to sellers, only if we have free 
	space in our stock, i.e.\ when the matching queue is not full.  
	We underestimate  the online profit by considering only selling an item to 
	the buyer that is matched to the seller from which we bought the item.
	Mimicking the analysis in the proof of Theorem~\ref{th:upperrrevenue1} 
	we can see that the expected number of items bought from the $\kappa$ 
	matched sellers is  
	$\kappa F_S(q)\geq \kappa\frac{1}{2eK}\frac{1}{r}$. 
	
	Now we argue that $q\leq \frac{\mu_B}{2}$. Indeed, since $F_S(q)\leq
	\frac{1}{e}$ we know for sure that $q\leq \mu_S$, and so from
	Lemma~\ref{lemma:tailconcavedistro} it is $q\leq e\mu_S F(q)\leq
	e\mu_S\frac{\mu_B}{\mu_S}\frac{1}{2e}=\frac{\mu_B}{2}$. Next, notice 
	that
	whenever we make a successful sale, the contribution to profit is 
	$p-q\geq
	\mu_B-\frac{\mu_B}{2}=\frac{1}{2}\mu_B$. 
	Thus, the total expected gain in profit from sales is at least 
	$$\kappa F_S(q) (1-F_B(p))(p-q)
	\geq 
	\kappa\frac{1}{2eK}\frac{1}{\frac{\mu_S}{\mu_B}+1}(1-F_B(\mu_B))\frac{1}{2}\mu_B
	\geq \frac{1}{4e^2}\frac{1}{Kr}\kappa\mu_B,$$ 
	where in the bound for the quantile $1-F_B(\mu_B)$ we used 
	Property~\ref{prop:MHR1} of Theorem~\ref{th:MHRprops}.
	Also, the profit we
	loose from the cost of unsold items cannot be more than $Kq\leq K\mu_S
	e\frac{1}{2eK}=O(\mu_S)$.
	On the other hand, the offline profit is at most $\kappa$ times the 
	expected 
	maximum order statistic out of $n$ independent draws from $F_B$, so 
	by 
	Property~\ref{prop:MHR2} of Theorem~\ref{th:MHRprops} it is upper 
	bounded by 
	$\kappa H_n\mu_B$. Putting everything together, the competitive ratio 
	of the online algorithm is at most 
	\begin{equation*}\frac{\kappa 
		H_n\mu_B}{\frac{1}{4e^2}\frac{1}{Kr}\kappa\mu_B}=O\left(Kr\ln n 
	\right).
	\end{equation*}
\end{proof}
\begin{remark}
	We want to mention here that the above upper bound in 
	Theorem~\ref{th:stocklimit},
	although a substantial improvement from the $\varTheta(\sqrt{n})$ one 
	for the general
	case in Theorem~\ref{th:upperrrevenue1}, it cannot be improved further: 
	the logarithmic
	lower bound is unavoidable, since a careful inspection of the welfare 
	lower bound in
	the proof of Lemma~\ref{th:welfarelower1} reveals that the same analysis 
	carries over to the
	profit. In particular, the last parenthesis of $\expect{Y\fwh{Y\geq \beta}} 
	- 
	\expect{Y\fwh{Y\geq \alpha}}$ in \eqref{eq:helper2} will be replaced by 
	$\beta-\alpha$ which is 
	still nonnegative, and also the bad instance sequence of $SB^n$ does not 
	use a stock of size more than $1$.
	We try to overcome this obstacles by considering a different model of 
	constrained streams in the following section. 
\end{remark}

\section{Balanced Sequences} \label{sec:balanced}
As we saw in Section~\ref{sec:stocklimit}, introducing a restriction in the 
size of available stock can improve the performance of our online 
algorithms with respect to profit. However, the bound is still 
super-constant. Thus, it is perhaps more reasonable to assume some 
knowledge of the ratio 
$\alpha$ between buyers and sellers in sequences the intermediary might 
face. This allows us finer control over the trade-off between high volume of 
trades and the hunt for greater order statistics.

In this section we analyse the competitive ratio for profit and welfare 
obtained by online 
algorithms on $\alpha$-\emph{balanced} sequences.
\begin{definition}\label{def:alpha-balance}
	Let $\alpha$ be a positive integer. A sequence containing $m$ buyers is 
	called
	$\alpha$-balanced if it contains $\alpha m$ sellers and the
	$i$-th buyer is preceded by at least $\alpha i$ sellers.
\end{definition}

For example, the sequence $SBSSBSBB$ is $1$-balanced, but 
$SBBSSB$ is not. Similarly, $SSSBSB$ is $2$-balanced, while
$SSBSBSSSB$ isn't. Note that since $n = n_S \frac{\alpha+1}{\alpha} = n_B 
(\alpha+1)$, we only need to know the number of buyers of a sequence. 
For convenience, we will denote it by $m$ instead of $n_B$, as it is used 
quite often.

\subsection{Profit}
We first work on profit, deriving bounds for a variety of online and offline 
mechanisms. Naturally,
there are two types of offline mechanisms: adaptive and
non-adaptive. The \emph{non-adaptive} posted-price mechanism
calculates all prices in advance based on the sequence of buyers and
sellers, while the \emph{adaptive} posted-price mechanism can alter
the prices on the fly, depending on the outcomes of previous trades.

We show that there is a competitive online mechanism for 
$\alpha$-balanced
sequences. To do this, we compare the optimal adaptive and
non-adaptive profit to the profit of a class of hypothetical
mechanisms, called \emph{fractional mechanisms}, which are allowed to
buy fractional quantities of items: posting the price $p$ would buy
exactly $F_S(p)$ items or sell $1-F_B(p)$ items. The advantage of using
fractional mechanisms is that at any point we know the exact quantity
of items in the hands of the intermediary instead of the expectation; an 
immediate consequence of this is that we know in advance
whether there is enough quantity to sell, which implies that \emph{the
	adaptive and non-adaptive versions of the optimal fractional
	mechanism are identical.}

We can now give an outline of the results in this section: For 
$\alpha$-balanced sequences $\sigma$ with $m$ buyers and $\alpha m$ 
sellers, we establish the following relations of optimal profits:
\begin{equation}\label{eq:cons-order}
\text{adaptive}(\sigma) \leq \text{fractional}(\sigma) 
\leq \text{fractional}(S^{\alpha m}B^m) 
\approx \text{non-adaptive}(\sigma),
\end{equation}
the last of which will be our online algorithm. We begin by the fractional 
offline mechanism.
\begin{theorem} \label{thm:fractional} The profit gained by
	the optimal fractional mechanism for the sequence $S^{\alpha m}B^m$ is
	\begin{equation} \label{eq:fractional_prices}
	\begin{aligned}
	& \max 
	& &  m\left(p(1-F_B(p)) - \alpha \cdot q F_S(q)\right) \\
	& \text{s.t.}  
	& & 1-F_B(p) = \alpha F_S(q) \\
	&&& p,q \in [0,\infty).
	\end{aligned}
	\end{equation}
\end{theorem}
\begin{proof}
	The profit and optimal prices can be calculated through the following 
	optimization:
	\begin{equation*}
	\begin{aligned}
	& \text{max} 
	& & \sum_{i=1}^m p_i(1-F_B(p_i)) - \sum_{i=1}^{\alpha m} q_i 
	F_S(q_i) \\
	& \text{s.t.}  
	& &\sum_{i=1}^m (1-F_B(p_i)) \le \sum_{i=1}^{\alpha m} F_S(q_i) \\
	&&& p_i,q_i \in [0,\infty),
	\end{aligned}
	\end{equation*}
	where $q_i$ and $p_i$ are the prices for buying and selling
	respectively.  However, we can assume that the first
	constraint is tight, as all $q_i$'s can be lowered until
	equality is achieved, without hurting the trades happening in
	the second half of the sequence. Remember, these are
	\emph{not} in expectation, but rather, fractions.
	
	This constrained optimization can be reduced to finding
	stationary points of its Lagrange function
	$$\mathcal{L} = \sum_{i=1}^m p_i(1-F_B(p_i)) - \sum_{i=1}^{\alpha m} 
	q_i 
	F_S(q_i) 
	- \lambda (\sum_{i=1}^m (1-F_B(p_i)) - \sum_{i=1}^{\alpha m} F_S(q_i)). 
	$$
	Taking its derivative with respect to price $p_i$ we get:
	\begin{align*}
	(1-F_B(p_i)) -p_i f_B(p_i) &= -\lambda f_B(p_i) & \Leftrightarrow && 
	p_i - \frac{1-F_B(p_i)}{f_B(p_i)} &= \lambda,
	\end{align*}
	which has at most one solution for any given $\lambda$ due to
	the distribution being regular. The treatment of $q_i$'s is
	similar, leading to a unique solution as well. Thus, since
	$p_i = p$ and $q_i = q$ for all $i$ we obtain the stated result.
\end{proof}

For other sequences containing $\alpha m$ sellers and $m$ buyers in a 
different order, we can use the following lemma to establish the middle part 
of inequality~\ref{eq:cons-order}.
\begin{lemma}\label{lemma:best-fraction}
	For any $\alpha$-balanced $\sigma$ with $m$ buyers, 
	$\text{fractional}(\sigma) \le \text{fractional}(S^{\alpha m}B^m)$ 
\end{lemma}
\begin{proof}
	Let $q_i$, $p_i$ be the prices set by the 
	optimal fractional mechanism for 
	sequence $\sigma$. These prices have to satisfy $\sum_1^m 
	(1-F_B(p_i)) \le \sum_1^{\alpha m} F_S(q_i)$, to ensure that the total 
	quantity 
	of items sold does not exceed the amount bought. Thus, the prices
	$p_i$, $q_i$ 
	represent a feasible solution to the optimization problem for the 
	sequence $S^{\alpha m}B^m$ and by definition, their profit is at most as 
	much as the optimal.
\end{proof}

\begin{restatable}{theorem}{adaptiveres}
	\label{thm:adaptive}
	For any sequence $\sigma$ we have $\text{adaptive}(\sigma) \le 
	\text{fractional}(\sigma)$.
\end{restatable}
The intuition behind the proof of the theorem is that the optimal
adaptive profit is bounded from above by the optimal fractional
adaptive profit (since fractional mechanisms is a more general class
of mechanisms); since in fractional mechanisms optimal adaptive and
non-adaptive profits are the same, the theorem follows. For a more
rigorous technical treatment, see Appendix~\ref{sec:appendix}.

At this point, we have a clear model of the adversary's power: the 
fractional mechanism's revenue for sequence $S^{\alpha m}B^m$, setting 
only two 
prices $p,q$ for sellers and buyers. Could we do the same online? It seems 
likely. After all, long sequences of buyers and sellers seem to lead to a 
similar amount of trading on average by a mechanism setting the same 
prices.

Based on the previous discussion we propose the following online posted 
price algorithm:

\begin{itemize}
	\item Use prices $p,q$ given by the optimal fractional solution for 
	$S^{\alpha m}B^m$(see Theorem~\ref{thm:fractional}).
\end{itemize}

This algorithm works without knowing the 
length of the sequence chosen by the adversary. 

\begin{lemma}\label{lemma:prefix-order}
	Let $A$ be the online algorithm defined by the optimal fractional offline 
	prices of \eqref{eq:fractional_prices}. Consider two 
	$\alpha$-balanced 
	sequences $\sigma_1$ and 
	$\sigma_2$ of equal 
	length. We write $\sigma_1 \succ \sigma_2$ whenever every prefix 
	of $\sigma_1$ contains more sellers than the prefix of $\sigma_2$ 
	having equal length. Then, 
	$\sigma_1 \succ \sigma_2 \Rightarrow\mathcal{R}(A,\sigma_1) \ge 
	\mathcal{R}(A,\sigma_2)$
\end{lemma}
\begin{proof}
	Assume the draws of $\sigma_1$ and $\sigma_2$ come from the same 
	probability space, so that the 
	$i$-th agent gets the same draw in both sequences. We will show that all 
	trades (or at least as many) that happened in $\sigma_2$ will occur in 
	$\sigma_1$. Let $i$ be the index of an arbitrary buyer 
	that was matched to a seller in $\sigma_2$ and $k$ the number of items 
	in stock when he arrives in $\sigma_1$. If $k>0$, then we trade with 
	him as we would 
	do in $\sigma_2$. If $k=0$, we have already traded at least as many 
	items as $\sigma_2$ at this point. To see this, note that since 
	$\sigma_1 \succ \sigma_2$, at least as many items have been bought 
	from the first $i-1$ agents of $\sigma_1$ than from $\sigma_2$ and 
	because $k=0$, at least as many have been traded. 
\end{proof}

Although not all sequences are comparable (e.g. $SSBBSB$ and $SBSSBB$), 
the sequence $(S^{\alpha} B)^m$ is the bottom element among all 
$\alpha$-balanced sequences of length $(\alpha+1)m$. This is trivial, as 
any balanced sequence must  
have at least $\ceil{\frac{i}{(\alpha+1)/(\alpha)}}$ sellers for any prefix 
of length $i$ and $(S^{\alpha}B)^m$ is tight for this bound.

To formalize our intuition of making the same number of trades in the long 
run, we reformulate our algorithm in the more familiar setting of random 
walks. Instead of considering agents separately, each ``timestep'' would be 
one sub-sequence $S^{\alpha}B$, giving $m$ steps in total. Thus, we are 
interested in the random
variables $Z_{i}$, denoting the items in stock at the end of each
step, starting with $Z_0 = 0$. Knowing the algorithm buys $\alpha m 
F_S(q)$ 
items in expectation, the
expected profit can be given by 
\begin{equation}\label{eqn:expected_profit}
\mathcal{R}((S^{\alpha}B)^m) = (\alpha mF_S(q) - \expect{Z_m})(p-q) -
\expect{Z_m}q,
\end{equation}
which is the revenue of the expected number of trades minus the cost
of the unsold items.

\begin{lemma}\label{lemma:bound_unsold}
	$\expect{Z_m} \le \sqrt{2m \alpha^2 \log m}\left(1 - \frac{2}{m}\right) 
	+ 
	2$
\end{lemma}
\begin{proof}
	The process $Z_i$ is almost a martingale but not quite:  clearly 
	$\expect{Z_i} \le \alpha m$ for all $i$ and we do have 
	$\expect{Z_{i+1}|Z_i \ge 1} = Z_i$ since the expected change 
	in items after that step is $\alpha F_S(q) - (1-F_B(p)) = 0$
	by Theorem~\ref{thm:fractional} . However, $\expect{Z_{i+1}|Z_i = 0} > 
	Z_i$,
	by the no short selling assumption.
	
	We can define $Y_{i}$ in the same probability space, where $Y_0=0$, 
	and
	\begin{equation}
	Y_{i+1} = Y_i +  
	\begin{cases}
	Z_{i+1} &\text{if } Y_i > 0\\
	-Z_{i+1} &\text{if } Y_i < 0\\
	\begin{cases}
	Z_{i+1} &\text{with probability } \frac{1}{2}\\
	-Z_{i+1} &\text{with probability } \frac{1}{2}
	\end{cases} &\text{if } Y_i = 0
	\end{cases}.
	\end{equation}
	The crucial observation is that $Y_i$ behaves similar to 
	$Z_i$ but has no barrier at 0. Notice, that $|Y_i| \ge Z_i$ for all $i$
	and $Y_i$ is a martingale.
	
	Moreover, we have that $|Y_{i+1} - Y_i| \le \alpha$ thus by the 
	Azuma-Hoeffding inequality we can bound the expected value
	$\expect{Z_m}$:
	\begin{align}\label{eqn:bound_unsold}
	\Pr[Z_m \ge x] &\le \Pr[|Y_m| \ge x] = \Pr[|Y_m - Y_0| \ge x] \le 
	2e^{\frac{-x^2}{2m\alpha^2}} \Rightarrow\\
	\expect{Z_m} &\le x\left(1 - 2e^{\frac{-x^2}{2m\alpha^2}}\right) + 
	2\alpha me^{\frac{-x^2}{2m\alpha^2}},
	\end{align}
	where we can set $x = \sqrt{2m \alpha^2 \log m}$ to obtain the simpler 
	form:
	\begin{equation}\label{ineq:bound_unsold}
	\expect{Z_m} \le \sqrt{2m \alpha^2 \log m}\left(1 - \frac{2}{m}\right) + 
	2\alpha.
	\end{equation}
\end{proof}

\begin{lemma}
	\label{lemma:PropsFractional}
	Let $r=\max\sset{2,\frac{\mu_S}{\mu_B}}$. The optimal value of
	Programme~\eqref{eq:fractional_prices} is at least 
	$m \frac{\mu_B}{2er}$.
	Furthermore, at any optimal solution the buyer price has to be at most 
	$p\leq4\ln(4er)\mu_B$.
\end{lemma}
\begin{proof}
	Consider the value of Programme~\eqref{eq:fractional_prices} that 
	corresponds to the
	solution determined by the seller price $q$ such that
	$ F_S(q)=\frac{1}{ e\alpha r}$. In a similar way to the proof of
	Theorem~\ref{th:upperrrevenue1}, it is again easy to see that $q\leq 
	\mu_S$ since 
	$F_S(q)\leq
	\frac{1}{e}$, and so by Lemma~\ref{lemma:tailconcavedistro} and the 
	regularity 
	of $F_S$ we
	get that $q\leq e\mu_S \frac{1}{e\alpha r}\leq\frac{\mu_S}{r} \frac{\mu_B}{2}$. Furthermore, 
	for
	the corresponding buyer price $p$ we have
	$1-F_B(p)=\alpha F_S(q)=\frac{1}{e
		r}<\frac{1}{e}$ and so from Property~\ref{prop:MHR1} of 
	Theorem~\ref{th:MHRprops} we get
	that $p\geq \mu_B$. Thus, the objective value of the particular solution 
	is at least
	$m\alpha F_S(q)(p-q)\geq m \frac{1}{e 
		r}(\mu_B-\frac{\mu_B}{2})=m
	\frac{\mu_B}{2er}$.
	
	Next, for the upper bound on the buyer price, consider a solution that 
	has buyer
	price $\hat p=cp^*$ for $c\geq 1$, where $F_B(p^*)=1-\frac{1}{e}$. Then, 
	since $F_B$ is an
	MHR distribution, $(1-F_B(x))^\frac{1}{x}$ is decreasing with respect to 
	$x$, 
	as can be
	verified using that $\log(1-F_B(x))$ is concave (see 
	e.g.~\cite{Barlow1964}), 
	so
	$
	1-F_B(\hat p)\leq (1-F_B(p^*))^{\frac{\hat p}{p^*}}=e^{-c}
	$.
	Furthermore, since $F_B(p^*)=1-\frac{1}{e}$, from 
	Property~\ref{prop:MHR12} of Theorem~\ref{th:MHRprops} it must be 
	that 
	$p^*\leq 2\mu_B$, and thus $\hat p\leq 2c\mu_B$, resulting in 
	$$1-F_B(2c\mu_B)\leq 1-F_B(\hat p)\leq e^{-c}.$$
	This means that if we use a solution with $p=2c\mu_B$, 
	for some
	$c\geq 1$, the objective value of the Programme cannot exceed
	$m (1-F_B(p))(p-q)\leq m e^{-c}2c\mu_B$. So, unless this value is at 
	least $m 
	\frac{\mu_B}{2er}$, the
	particular choice of $p$ cannot be part of an optimal solution. Thus, it 
	must be
	$ce^{-c}\geq \frac{1}{4er}$. It is not difficult to check that this requires 
	$c\leq 2\ln(4er)$, since $2\ln xe^{-2\ln x}=\frac{2\ln 
		x}{x^2}<\frac{1}{x}$ for any $x>0$ and $ce^{-c}$ is a decreasing function for $c\geq 1$.
	As a result of the above analysis we can conclude that the buyer price $p$ of any optimal solution in Programme~\eqref{eq:fractional_prices} must be such that $p<2\mu_B$, or otherwise satisfy $p\leq 2\cdot 2\ln(4er)\cdot\mu_B=4\ln(4er)\mu_B$. In any case, the desired upper bound for $p$ in the theorem's statement holds.
\end{proof}

\begin{theorem}
	Under our standard regularity assumptions, the proposed non-adaptive 
	online mechanism is $(1+o(\alpha^{3/2} r\log r))$-competitive for any 
	balanced 
	sequence, where $r=\max\sset{2,\frac{\mu_S}{\mu_B}}$.
\end{theorem}
\begin{proof}
	Plugging \eqref{ineq:bound_unsold} into \eqref{eqn:expected_profit}, we 
	get:
	\begin{align}
	\mathcal{R}((S^{\alpha}B)^m) 
	&\ge \alpha mF_S(q)(p-q) - \expect{Z_m}(p-q) - \expect{Z_m}q\notag\\
	&\ge \alpha mF_S(q)(p-q) - \left(\sqrt{2m \alpha^2 \log m}\left(1 - 
	\frac{2}{m}\right) 
	+ 2\alpha\right)p \notag\\
	&\ge \alpha mF_S(q)(p-q) - O(\alpha \sqrt{m\ln m}p) \label{eq:helper1}.
	\end{align}
	Using Lemma~\ref{lemma:best-fraction}, Theorem~\ref{thm:adaptive} 
	and 
	Theorem~\ref{thm:fractional} we know that for every $\alpha$-balanced 
	sequence, the profit
	of our non-adaptive online algorithm is at least 
	$\mathcal{R}((S^{\alpha}B)^m) $
	and the optimal offline is at most that of the fractional on sequence 
	$S^{\alpha
		m}B^m$, i.e.\ $\alpha mF_S(q)(p-q)$. 
	Thus, the second term in \eqref{eq:helper1} bounds the additive 
	difference 
	of the online and
	optimal offline profit, and its ratio with respect to the offline profit  is 
	upper bounded by
	$$
	O\left( \frac{\alpha \sqrt{m\ln m}p}{\alpha mF_S(q)(p-q)} \right)
	=O\left( \frac{\alpha \sqrt{m\ln m} \mu_B \ln (4er) 
	}{m \frac{\mu_B}{2er}} \right)
	=O\left(\alpha^{3/2} \sqrt{\frac{\ln n}{n}}r\log r 
	\right)=o(\alpha^{3/2} r\log 
	r),
	$$
	using $m = n/(\alpha+1)$.

\end{proof}

\begin{remark}
	Among all
	$1$-balanced sequences, the sequence that gives the maximum profit is 
	not
	the sequence $S^mB^m$; intuitively, by moving some buyers earlier in
	the sequence, we obtain an improved profit by adapting the remaining
	buying prices to the outcome of these potential trades. For example,
	it should be intuitively clear that the sequence
	$S^{m/2}BS^{m/2}B^{m-1}$ has (slightly) better adaptive profit than
	the sequence $S^mB^m$ for large $m$. Our work above shows that the
	difference is asymptotically insignificant, but it remains an
	intriguing question to determine the balanced sequence with the
	maximum profit.
\end{remark}

\subsection{Welfare}
Welfare on balanced sequences also improves the competitive ratio of 
Theorem~\ref{th:welfareupper2} to a constant. Intuitively, the reason is that 
the high volume of possible trades dampens the advantage the adversary 
has in obtaining higher order statistics from buyers. As before, the fact that 
all sellers start with some contribution to the welfare is also helpful.

\begin{theorem}\label{thm:cons-welfare}
	The online auction that 
	posts to any seller and buyer the median of their distribution is 
	$4$-competitive. 
\end{theorem}
\begin{proof}
	The algorithm buys from half the sellers in expectation, so in the end the 
	welfare obtained just from sellers is at least:
	\begin{equation*}
	\expect{\sum_{t\in N_S \setminus I_S} X_t} 
	= \sum_{t \in N_S} \expect{X_t | X_t \ge q}(1 - F_S(q))
	\ge \frac{1}{2} n_S \mu_S.
	\end{equation*}
	Following the proof of Lemma~\ref{th:welfareupper1}, let $\kappa$ 
	denote the size the matching between sellers and buyers. Since the input 
	is $\alpha$-balanced, we are guaranteed that every buyer is preceded by 
	some \emph{distinct} seller, meaning that $\kappa$ is exactly $N_B$. 
	The welfare obtained from buyers is
	\begin{equation*}
	\kappa \Pr{[X_S \le q]} \Pr{[X_B \ge p]} \expect{X_B | X_B \ge p} \ge 
	\frac{1}{4}n_B \mu_B,
	\end{equation*}
	Adding everything together, the online 
	algorithm gets at least $\frac{1}{4}(n_B \mu_B + n_S \mu_S)$. On the 
	other hand the optimal welfare is 
	at most:
	\begin{equation*}
	\expect{\sum_{t\in N_S \setminus I_S} X_t + \sum_{t\in I_B} X_t} \\
	\le \expect{\sum_{t\in N_S} X_t + \sum_{t\in N_B} X_t} \\ 
	= n_S \mu_s + n_B \mu_B.
	\end{equation*}
\end{proof}	
Notice that the above theorem holds without any regularity assumption on 
the 
agent value distributions.

\subparagraph*{Acknowledgements.}
We want to thank Matthias Gerstgrasser for many helpful discussions and 
his assistance during the initial development of our paper.

\nocite{niazadeh2014simple}
\bibliography{online_auction}

\begin{thebibliography}{10}

\bibitem{Arnold1979}
Barry~C. Arnold and Richard~A. Groeneveld.
\newblock {Bounds on Expectations of Linear Systematic Statistics Based on
  Dependent Samples}.
\newblock {\em The Annals of Statistics}, 7(1):220--223, jan 1979.
\newblock URL: \url{http://projecteuclid.org/euclid.aos/1176344567}, \href
  {http://dx.doi.org/10.1214/aos/1176344567}
  {\path{doi:10.1214/aos/1176344567}}.

\bibitem{Babaioff2011a}
Moshe Babaioff, Liad Blumrosen, Shaddin Dughmi, and Yaron Singer.
\newblock {Posting Prices with Unknown Distributions}.
\newblock In {\em Innovations in Computer Science (ICS)}, jan 2011.
\newblock URL:
  \url{http://conference.itcs.tsinghua.edu.cn/ICS2011/content/paper/35.pdf}.

\bibitem{BabaioffDKS15}
Moshe Babaioff, Shaddin Dughmi, Robert~D. Kleinberg, and Aleksandrs Slivkins.
\newblock Dynamic pricing with limited supply.
\newblock {\em {ACM} Trans. Economics and Comput.}, 3(1):4, 2015.
\newblock URL: \url{http://doi.acm.org/10.1145/2559152}, \href
  {http://dx.doi.org/10.1145/2559152} {\path{doi:10.1145/2559152}}.

\bibitem{Babaioff_secretary_immorlica}
Moshe Babaioff, Nicole Immorlica, David Kempe, and Robert Kleinberg.
\newblock Online auctions and generalized secretary problems.
\newblock {\em SIGecom Exch.}, 7(2):7:1--7:11, June 2008.
\newblock URL: \url{http://doi.acm.org/10.1145/1399589.1399596}, \href
  {http://dx.doi.org/10.1145/1399589.1399596}
  {\path{doi:10.1145/1399589.1399596}}.

\bibitem{badanidiyuru_learning_2012}
Ashwinkumar Badanidiyuru, Robert Kleinberg, and Yaron Singer.
\newblock Learning on a budget: posted price mechanisms for online procurement.
\newblock In {\em Proceedings of the 13th {ACM} {Conference} on {Electronic}
  {Commerce}}, pages 128--145. ACM, 2012.
\newblock URL: \url{http://dl.acm.org/citation.cfm?id=2229026}.

\bibitem{Bar-Yossef:2002_online}
Ziv Bar-Yossef, Kirsten Hildrum, and Felix Wu.
\newblock Incentive-compatible online auctions for digital goods.
\newblock In {\em Proceedings of the Thirteenth Annual ACM-SIAM Symposium on
  Discrete Algorithms}, SODA '02, pages 964--970, Philadelphia, PA, USA, 2002.
  Society for Industrial and Applied Mathematics.
\newblock URL: \url{http://dl.acm.org/citation.cfm?id=545381.545506}.

\bibitem{Barlow1964}
Richard~E. Barlow and Albert~W. Marshall.
\newblock {Bounds for Distributions with Monotone Hazard Rate, I}.
\newblock {\em The Annals of Mathematical Statistics}, 35(3):1234--1257, sep
  1964.
\newblock URL: \url{http://projecteuclid.org/euclid.aoms/1177703281}, \href
  {http://dx.doi.org/10.1214/aoms/1177703281}
  {\path{doi:10.1214/aoms/1177703281}}.

\bibitem{blum_near-optimal_2005}
Avrim Blum and Jason~D. Hartline.
\newblock Near-optimal online auctions.
\newblock In {\em Proceedings of the sixteenth annual {ACM}-{SIAM} symposium on
  {Discrete} algorithms}, pages 1156--1163. Society for Industrial and Applied
  Mathematics, 2005.
\newblock URL: \url{http://dl.acm.org/citation.cfm?id=1070597}.

\bibitem{blum_online_2006}
Avrim Blum, Tuomas Sandholm, and Martin Zinkevich.
\newblock Online algorithms for market clearing.
\newblock {\em Journal of the ACM (JACM)}, 53(5):845--879, 2006.
\newblock URL: \url{http://dl.acm.org/citation.cfm?id=1183913}.

\bibitem{blumrosen_almost_2016}
Liad Blumrosen and Shahar Dobzinski.
\newblock ({Almost}) {Efficient} {Mechanisms} for {Bilateral} {Trading}.
\newblock {\em arXiv preprint arXiv:1604.04876}, 2016.
\newblock URL: \url{http://arxiv.org/abs/1604.04876}.

\bibitem{Blumrosen2008}
Liad Blumrosen and Thomas Holenstein.
\newblock Posted prices vs. negotiations: An asymptotic analysis.
\newblock In {\em Proceedings of the 9th ACM Conference on Electronic
  Commerce}, EC '08, pages 49--49, New York, NY, USA, 2008. ACM.
\newblock URL: \url{http://doi.acm.org/10.1145/1386790.1386801}, \href
  {http://dx.doi.org/10.1145/1386790.1386801}
  {\path{doi:10.1145/1386790.1386801}}.

\bibitem{blumrosen_approximating_2016}
Liad Blumrosen and Yehonatan Mizrahi.
\newblock Approximating {Gains}-from-{Trade} in {Bilateral} {Trading}.
\newblock In {\em Web and {Internet} {Economics}}, pages 400--413. Springer,
  Berlin, Heidelberg, December 2016.
\newblock URL:
  \url{http://link.springer.com/chapter/10.1007/978-3-662-54110-4_28}.

\bibitem{Borodin1998a}
Allan Borodin and Ran El-Yaniv.
\newblock {\em {Online Computation and Competitive Analysis}}.
\newblock Cambridge University Press, 1998.

\bibitem{chawla_multi-parameter_2010}
Shuchi Chawla, Jason~D. Hartline, David~L. Malec, and Balasubramanian Sivan.
\newblock Multi-parameter mechanism design and sequential posted pricing.
\newblock In {\em Proceedings of the forty-second {ACM} symposium on {Theory}
  of computing}, pages 311--320. ACM, 2010.
\newblock URL: \url{http://dl.acm.org/citation.cfm?id=1806733}.

\bibitem{colini2016approximately}
Riccardo Colini-Baldeschi, Bart de~Keijzer, Stefano Leonardi, and Stefano
  Turchetta.
\newblock Approximately efficient double auctions with strong budget balance.
\newblock In {\em Proceedings of the 27th Annual ACM-SIAM Symposium on Discrete
  Algorithms}, pages 1424--1443. Society for Industrial and Applied
  Mathematics, 2016.

\bibitem{Deng2014}
X~Deng, P~Goldberg, B~Tang, and J~Zhang.
\newblock {Revenue maximization in a bayesian double auction market}.
\newblock {\em Theoretical Computer Science}, 2014.
\newblock URL:
  \url{http://www.sciencedirect.com/science/article/pii/S0304397514002928}.

\bibitem{feldman2015combinatorial}
Michal Feldman, Nick Gravin, and Brendan Lucier.
\newblock Combinatorial auctions via posted prices.
\newblock In {\em Proceedings of the Twenty-Sixth Annual ACM-SIAM Symposium on
  Discrete Algorithms}, pages 123--135. Society for Industrial and Applied
  Mathematics, 2015.

\bibitem{gerstgrasser2016revenue}
Matthias Gerstgrasser, Paul~W Goldberg, and Elias Koutsoupias.
\newblock Revenue maximization for market intermediation with correlated
  priors.
\newblock In {\em International Symposium on Algorithmic Game Theory}, pages
  273--285. Springer, 2016.

\bibitem{gkyr2015-wine}
Yiannis Giannakopoulos and Maria Kyropoulou.
\newblock {The VCG Mechanism for Bayesian Scheduling}.
\newblock In {\em Proceedings of the 11th Conference on Web and Internet
  Economics}, WINE'15, pages 343--356. 2015.
\newblock URL: \url{http://arxiv.org/abs/1509.07455}, \href
  {http://dx.doi.org/10.1007/978-3-662-48995-6_25}
  {\path{doi:10.1007/978-3-662-48995-6_25}}.

\bibitem{hajiaghayi_online_2005}
Mohammad~T. Hajiaghayi.
\newblock Online auctions with re-usable goods.
\newblock In {\em Proceedings of the 6th {ACM} conference on {Electronic}
  commerce}, pages 165--174. ACM, 2005.
\newblock URL: \url{http://dl.acm.org/citation.cfm?id=1064027}.

\bibitem{hajiaghayi_adaptive_2004}
Mohammad~Taghi Hajiaghayi, Robert Kleinberg, and David~C. Parkes.
\newblock Adaptive limited-supply online auctions.
\newblock In {\em Proceedings of the 5th {ACM} conference on {Electronic}
  commerce}, pages 71--80. ACM, 2004.
\newblock URL: \url{http://dl.acm.org/citation.cfm?id=988784}.

\bibitem{hajiaghayi2007automated}
Mohammad~Taghi Hajiaghayi, Robert Kleinberg, and Tuomas Sandholm.
\newblock Automated online mechanism design and prophet inequalities.
\newblock In {\em AAAI}, volume~7, pages 58--65, 2007.

\bibitem{Kleinberg:2003_online_learning}
Robert Kleinberg and Tom Leighton.
\newblock The value of knowing a demand curve: Bounds on regret for online
  posted-price auctions.
\newblock In {\em Proceedings of the 44th Annual IEEE Symposium on Foundations
  of Computer Science}, FOCS '03, pages 594--, Washington, DC, USA, 2003. IEEE
  Computer Society.
\newblock URL: \url{http://dl.acm.org/citation.cfm?id=946243.946352}.

\bibitem{Kleinberg:2012}
Robert Kleinberg and Seth~Matthew Weinberg.
\newblock {Matroid Prophet Inequalities}.
\newblock In {\em Proceedings of the Forty-fourth Annual ACM Symposium on
  Theory of Computing}, STOC '12, pages 123--136, New York, NY, USA, 2012. ACM.
\newblock URL:
  \url{http://doi.acm.org.eaccess.ub.tum.de/10.1145/2213977.2213991}.

\bibitem{koutsoupias2013competitive}
Elias Koutsoupias and George Pierrakos.
\newblock {On the competitive ratio of online sampling auctions}.
\newblock {\em ACM Transactions on Economics and Computation}, 1(2):10, May
  2013.
\newblock URL: \url{http://dl.acm.org/citation.cfm?id=2465769.2465775}.

\bibitem{mcafee_dominant_1992}
R.~Preston McAfee.
\newblock A dominant strategy double auction.
\newblock {\em Journal of Economic Theory}, 56(2):434--450, 1992.
\newblock URL:
  \url{http://www.sciencedirect.com/science/article/pii/002205319290091U}.

\bibitem{myerson1983efficient}
Roger~B Myerson and Mark~A Satterthwaite.
\newblock Efficient mechanisms for bilateral trading.
\newblock {\em Journal of Economic Theory}, 29(2):265--281, 1983.

\bibitem{niazadeh2014simple}
Rad Niazadeh, Yang Yuan, and Robert Kleinberg.
\newblock Simple and near-optimal mechanisms for market intermediation.
\newblock In {\em International Conference on Web and Internet Economics},
  pages 386--399. Springer, 2014.

\bibitem{Nisan:2007:AGT}
David~C. Parkes.
\newblock Online mechanisms.
\newblock In Noam Nisan, Tim Roughgarden, Eva Tardos, and Vijay~V. Vazirani,
  editors, {\em Algorithmic Game Theory}, chapter~16. Cambridge University
  Press, New York, NY, USA, 2007.

\bibitem{Yan2011}
Qiqi Yan.
\newblock Mechanism design via correlation gap.
\newblock In {\em Proceedings of the Twenty-second Annual ACM-SIAM Symposium on
  Discrete Algorithms}, SODA '11, pages 710--719. SIAM, 2011.
\newblock URL: \url{http://dl.acm.org/citation.cfm?id=2133036.2133092}.

\end{thebibliography}

\appendix
\section{Omitted Proofs}
\label{sec:appendix}

\begin{lemma}\label{lemma:fifo}  
	The matching computed using an online FIFO queue of size $K$, adding 
	sellers while it's not full and popping them when a buyer is encountered,  
	in the proof of Theorem~\ref{th:stocklimit} is a maximum one.
\end{lemma}
\begin{proof}
	We show this for the limited stock case. The general case
	works similarly, or follows by setting $K$ large enough. Let 
	$\fifo$ be
	the matching computed by our FIFO algorithm, and let 
	$\optm$ be 
	any
	arbitrary maximum matching in the graph induced by $\sigma$. We will 
	show that we can
	transform $\optm$ into $\fifo$ using a series of changes 
	that do not reduce
	its size.
	
	Let $i$ be the index of the first vertex that is not matched in the
	same way in $\fifo$ and $\optm$. That is, all edges in 
	$\fifo$ and $\optm$ that are
	either between vertices before $i$, or originate at a vertex before
	$i$, are identical in both matchings. (There cannot be any matchings
	that terminate in a vertex smaller than $i$ but originate after $i$
	due to the construction of the graph.) We will show using a
	case-by-case analysis that we can change $\optm$ into 
	$\optm'$ so that $i$ is
	matched the same way as in $\fifo$, without changing any edges 
	originating
	before $i$, and with $|\optm| = |\optm'|$. It follows that we 
	can repeat this
	procedure until $\optm$ is transformed into $\fifo$, and 
	thus $|\optm| =|\fifo|$,
	i.e. $\fifo$ is a maximum matching.
	
	\begin{enumerate}
		\item If $i$ is a buyer: This is not possible. If $i$ is matched in
		either matching, the edge is originating from a vertex before $i$, and
		thus must be the same in both matchings by our hypothesis.
		\item If $i$ is a seller.
		\begin{enumerate}
			\item If $i$ is matched in both matchings. Let $j_\fifo$ be 
			its match in
			$\fifo$, and $j_\optm$ in $\optm$.
			\begin{enumerate}
				\item $j_\fifo < j_\optm$
				\begin{enumerate}
					\item $j_{\fifo}$ unmatched in $\optm$. Make 
					edge $ij_{\optm}$ into $ij_{\fifo}$. Can't
					violate $K$-limit this way, as we're making the edge shorter.
					\item $j_{\fifo}$ matched in $\optm$. Make edge 
					$ij_{\optm}$ into $ij_{\fifo}$, and match
					the seller originally matched to $j_{\fifo}$ in 
					$\optm$ with $j_{\optm}$. We can't
					violate the $K$-limit this way.
				\end{enumerate}
				\item $j_{\optm} < j_{\fifo}$ - This is not possible.
				\begin{enumerate}
					\item It is not possible that $j_{\optm}$ is unmatched in 
					$\fifo$, as we
					encounter it before $j_{\fifo}$, and would have matched 
					$i$ to it.
					\item It is not possible that $j_{\optm}$ is matched to a 
					seller other than
					$i$ in $\fifo$. Not to one before $i$ by hypothesis, and 
					not to one after
					$i$ by construction of the FIFO algorithm.
				\end{enumerate}
			\end{enumerate}
			\item If $i$ is matched in $\fifo$ but not in $\optm$. 
			Let $j_{\fifo}$ be $i$'s
			match in $\optm$.
			\begin{enumerate}
				\item $j_{\fifo}$ unmatched in $\optm$. This cannot 
				happen. Notice that we
				cannot have any buyers between $i$ and $j_{\fifo}$ that 
				are unmatched in
				$\fifo$, nor can we have any that are matched to sellers 
				after $i$. Thus,
				all buyers between $i$ and $j_{\fifo}$ are matched to 
				sellers before $i$ in
				both $\fifo$ and $\optm$. There can be at most 
				$K-1$ of them, as there is one
				more edge originating from $i$ in $\fifo$, and a cut 
				between $i$ and $i+1$
				has size at most $K$ in $\fifo$. Therefore, we could add 
				the edge $ij_{\fifo}$
				to $\optm$ without violating the $K$-limit. Thus 
				$\optm$ was not maximum,
				contradicting out assumption.
				\item $j_{\fifo}$ matched in $\optm$. Let 
				$s_{\optm}$ be the seller matched to $j_{\fifo}$
				in $\optm$. Again, all buyers between $i$ and 
				$j_{\fifo}$ are matched to sellers
				before $i$ in both matchings. So we can replace $s_{\optm} 
				j_{\fifo}$ with $sj_{\fifo}$
				without violating the $K$- limit.
			\end{enumerate}
			\item If $i$ is matched in $\optm$ but not in $\fifo$. 
			Let $j_{\optm}$ be its match
			in $\optm$.
			\begin{enumerate}
				\item $j_{\optm}$ is matched in $\fifo$. This cannot
				happen due to the FIFO construction.
				\item $j_{\optm}$ is unmatched in $\fifo$. This 
				cannot happen. If $i$ were to
				enter the FIFO queue, it would be matched to $j_{\optm}$ 
				(or an earlier
				available buyer) in $\fifo$. If $i$ does not enter the FIFO 
				queue this can
				only be because the queue was full.  But if the queue was full, this
				means that $K$ sellers before $i$ were matched to buyers 
				between $i$
				and $j_{\optm}$ (otherwise $j_{\optm}$ would be 
				matched to one of them if $\fifo$). So
				there is $K$ edges going from sellers before $i$ to buyers 
				between $i$
				and $j_{\optm}$ in $\fifo$. So there is also $K$ 
				edges going that way in $\optm$, as
				they are identical on vertices before $i$. So there is $K+1$ edges
				going from nodes before and including $i$ to vertices after $i$ in
				$\optm$, violating the $K$ item limit.
			\end{enumerate}
		\end{enumerate}
		
	\end{enumerate}
\end{proof}

\adaptiveres*
\begin{proof}
  Fix an adaptive mechanism and let $Q_i$ be the price posted to
  seller $i$ and $\tilde{Q_i}$ be the probability of sale at price
  $Q_i$. Since in an adaptive mechanism the price depends on the
  history, $Q_i$ and $\tilde{Q_i}$ are random variables. Similarly
  define $P_j$ and $\tilde{P_j}$ to be the price and probability of
  buying from buyer $j$. For the payments to sellers and from buyers we 
  have:
  \begin{align*}
    \expect{Q_i F_S(Q_i)} & = \expect{\tilde{Q_i} F_S^{-1}(\tilde{Q_i})} \\
    \expect{P_j(1-F_B(P_j))} & = \expect{\tilde{P_j} F_B^{-1}(1-\tilde{P_j})}.
  \end{align*}
  Summing over all agents we get the expected profit:
	\begin{align}
	&\sum_{j \in N_B}\expect{P_j(1-F_B(P_j))} -
          \sum_{i \in N_S}\expect{Q_i F_S(Q_i)} \nonumber \\
	&= \sum_{j \in N_B} \expect{\tilde{P_j} F_B^{-1}(1-\tilde{P_j})} - 
	\sum_{i \in N_S} \expect{\tilde{Q_i} F_S^{-1}(\tilde{Q_i})} 
	\nonumber \\
	&\le \sum_{j \in N_B} \expect{\tilde{P_j}} F_B^{-1}(1-\expect{\tilde{P_j}}) 
	- 
	\sum_{i \in N_S} \expect{\tilde{Q_i}}
          F_S^{-1}(\expect{\tilde{Q_i}}),  \label{ineq:regular}
	\end{align}
	where the last inequality follows from our regularity assumptions. Note 
	that in the last
        inequality $F_B^{-1}(1-\expect{\tilde{P_j}})$ and
        $F_S^{-1}(\expect{\tilde{Q_i}}) $ can be interpreted as prices
        set by the fractional mechanism, with $\expect{\tilde{P_j}}$
        and $\expect{\tilde{Q_i}}$ the fractions of items bought and
        sold.
	
	We have obtained the objective function of the optimization
        and it is left to a set of inequalities concerning the prices,
        to serve as the constraints.  Observe that
        $\expect{\tilde{Q_i}}$ is the expected number of items bought
        from seller $i$, while $\expect{\tilde{P_j}}$ sold to buyer
        $j$. Let $\mathcal{S}_t$ and $\mathcal{B}_t$ be the sets of
        indices of sellers and buyers contained in the first $t$
        agents of the sequence.
	
	Let $Z_t$ be the number of items exchanged with the agent encountered 
	at step $t$. The number of items currently held by the intermediary at 
	time $t$ is $\sum_1^t Z_i \ge 0$ by the no short selling assumption. 
	Thus for all $t$:
	\begin{align}
	\expect{\sum_{i=1}^t Z_i} &= \sum_{i \in \mathcal{S}_t} \expect{Z_i}  
	- 
	\sum_{j \in 
		\mathcal{B}_t} \expect{Z_j} \nonumber \\
	&=  \sum_{i \in \mathcal{S}_t} \expect{\expect{Z_i|\tilde{Q_i}}}  - 
	\sum_{j \in \mathcal{B}_t} \expect{\expect{Z_j|\tilde{P_j}}}
          \nonumber \\
	&=  \sum_{i \in \mathcal{S}_t} \expect{\tilde{Q_i}}  - 
	\sum_{j \in \mathcal{B}_t} \expect{\tilde{P_j}} \ge 0 \label{ineq:cons}
	\end{align}
	
	Combining \eqref{ineq:regular} and \eqref{ineq:cons} gives us
        exactly the same optimization problem the optimal fractional
        mechanism would face for that sequence.
\end{proof}

\end{document}